\documentclass[]{interact}

\usepackage{epstopdf}
\usepackage[caption=false]{subfig}

\usepackage[numbers,sort&compress]{natbib}
\bibpunct[, ]{[}{]}{,}{n}{,}{,}
\renewcommand\bibfont{\fontsize{10}{12}\selectfont}

\theoremstyle{plain}
\newtheorem{theorem}{Theorem}[section]
\newtheorem{lemma}[theorem]{Lemma}

\theoremstyle{definition}
\newtheorem{definition}[theorem]{Definition}

\theoremstyle{remark}

\newcommand{\Expect}{\mathsf{E}}

\begin{document}


\title{Modeling and Quickest Detection of a Rapidly Approaching Object}

\author{
	\name{Tim Brucks\textsuperscript{a}, Taposh Banerjee\textsuperscript{b}\thanks{CONTACT Taposh Banerjee. Email: taposh.banerjee@pitt.edu}, and Rahul Mishra\textsuperscript{c}}
	\affil{\textsuperscript{a}University of Texas at San Antonio; \textsuperscript{b}University of Pittsburgh; \textsuperscript{c}U. R. Rao Satellite Center, Bangalore, India.}
}
\maketitle

\begin{abstract}
	The problem of detecting the presence of a signal that can lead to a disaster is studied. A decision-maker collects data sequentially over time. At some point in time, called the change point, the distribution of data changes. This change in distribution could be due to an event or a sudden arrival of an enemy object. If not detected quickly, this change has the potential to cause a major disaster. In space and military applications, the values of the measurements can stochastically grow with time as the enemy object moves closer to the target. 
A new class of stochastic processes, called exploding processes, is introduced to model stochastically growing data. An algorithm is proposed and shown to be asymptotically optimal as the mean time to a false alarm goes to infinity. 
\end{abstract}


\section{Introduction}
In the problem of quickest change detection (QCD), a decision maker collects a sequence of measurements. The values in the sequence are seen as a realization of a stochastic process. It is assumed that the law of this stochastic process initially follows a distribution that is believed to be normal. At some point in time, called the change point, the statistical properties of the measurements or the law of the process changes. The goal of the QCD problem is to detect this change in distribution as quickly as possible while avoiding many false alarms \cite{veer-bane-elsevierbook-2013, poor-hadj-qcd-book-2009, tart-niki-bass-2014}. The most common change point model is the abrupt and persistent change model \cite{mous-astat-1986, shir-opt-stop-book-1978, lai-ieeetit-1998, tart-veer-siamtpa-2005, tartakovsky2017asymptotic, Pergamenchtchikov2018}. In this model,the law of the process abruptly changes at the change point and persists with that new law forever. For example, in Fig.~\ref{fig:abcpmodel} (left), we have plotted the mean values of a sequence of Gaussian random variables. At the change point ($80$ in the figure), the mean abruptly changes from $0$ to $4$ and stays at $4$ forever.  
\begin{figure}
	\centering
	\includegraphics[scale=0.4]{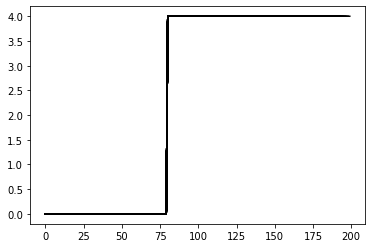} \quad 
		\includegraphics[scale=0.4]{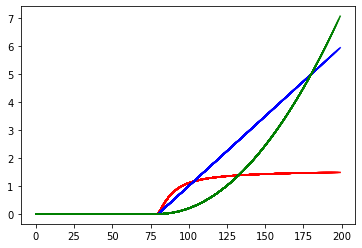}
	\caption{Left figure: Classical abrupt change point model. Figure shows the mean of a sequence of Gaussian random variables. This mean is $0$ before change and is $4$ after change. Right figure: In some space or military applications, the post-change measurements can have an exploding effect (linear, super-linear, or sublinear). The values of measurements can stochastically increase as the enemy object comes closer.}
	\label{fig:abcpmodel}
\end{figure}

In some military and space applications, the post-change law can have an exploding nature. 
\begin{enumerate}
	\item Space scientists are concerned with the detection of debris or other hazardous objects approaching a satellite and destroying it 
	(see Fig.~\ref{fig:satellite}). As the target object rapidly approaches the satellite, the mean of the measurements will have an exploding nature as shown in Fig.~\ref{fig:abcpmodel} (right). In Section~\ref{sec:satellite}, we provide a detailed discussion on this application. 
	\item Another classical example is enemy object detection in military applications. As a missile approaches a target or a torpedo approaches a submarine or ship, the values of the measurements are expected to increase with time. 
\end{enumerate}

\begin{figure}
	\centering
	\includegraphics[scale=0.1]{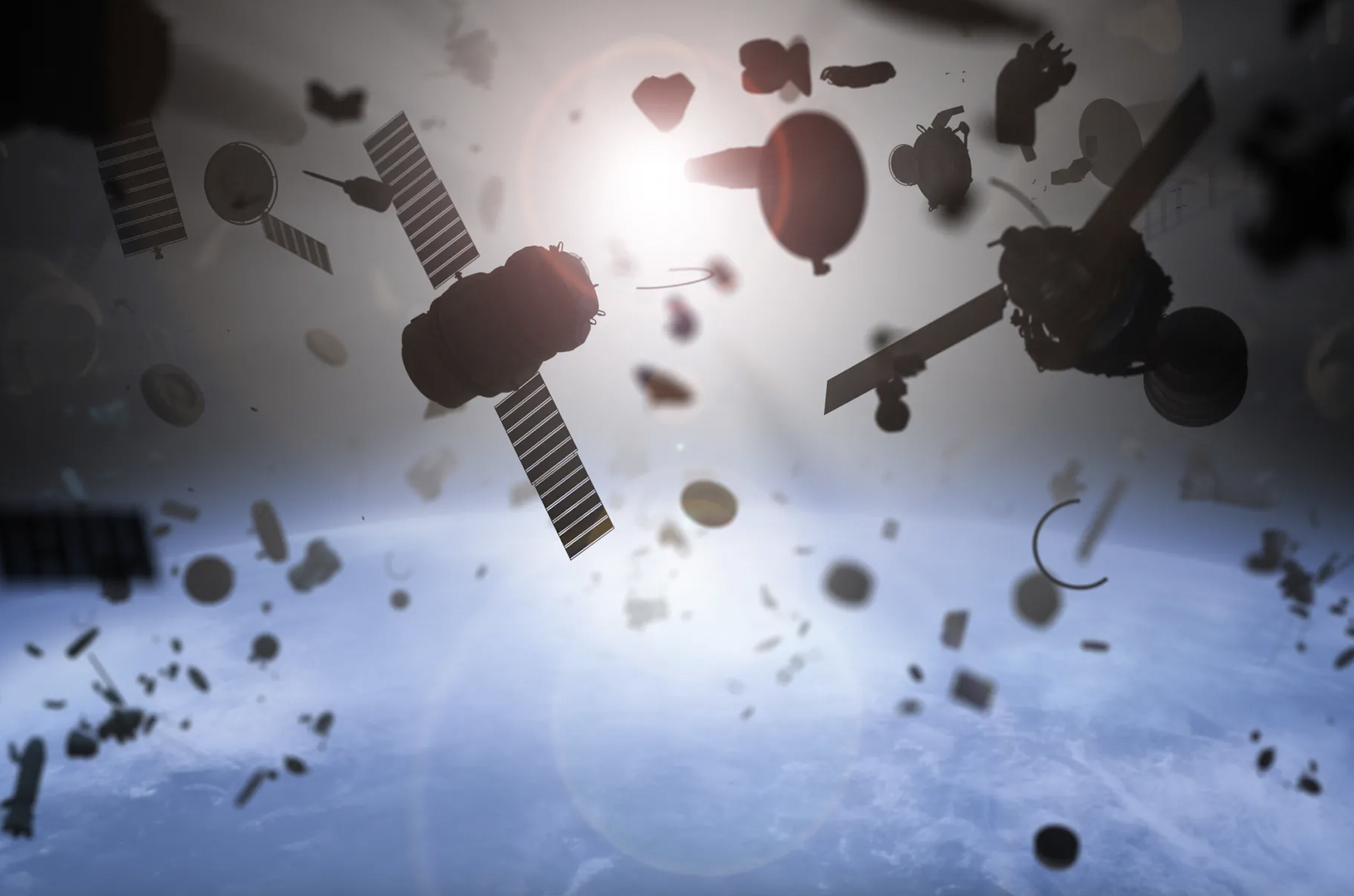}
	\includegraphics[scale=0.12]{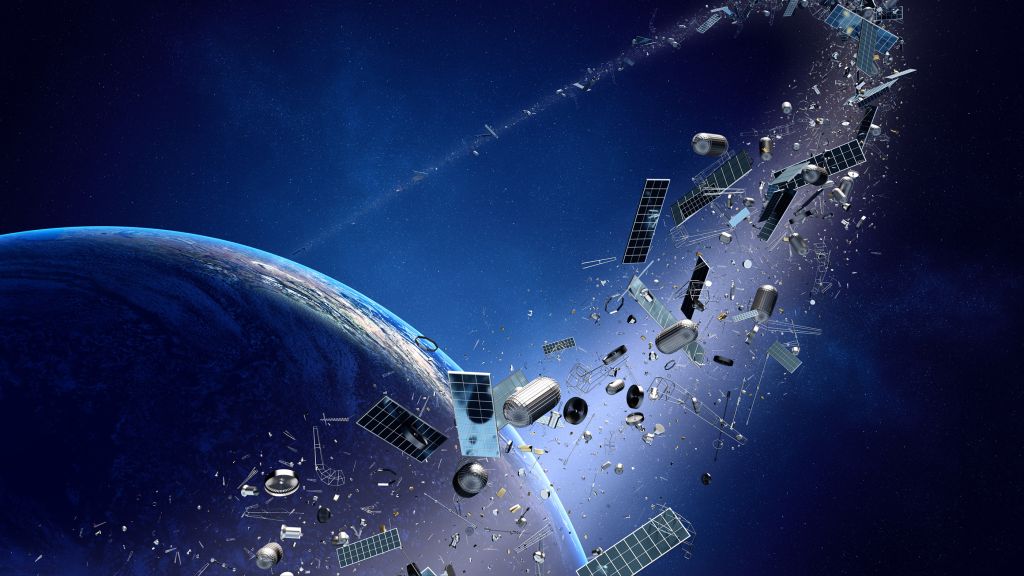}\\
	\includegraphics[scale=0.19]{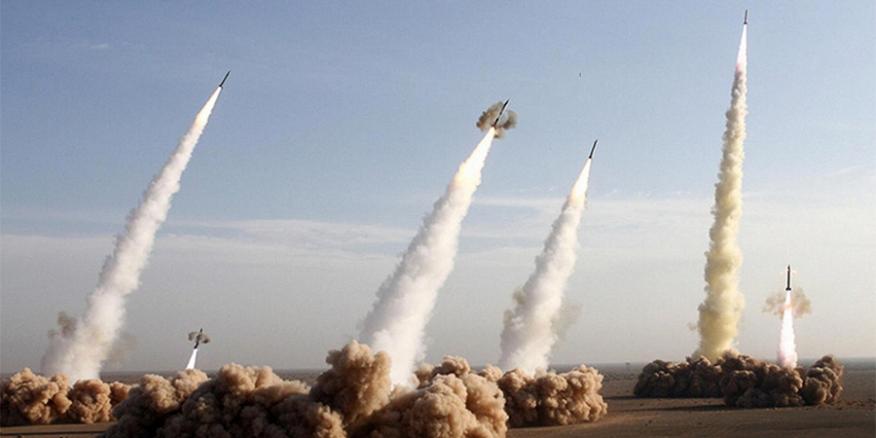}
	\includegraphics[scale=0.18]{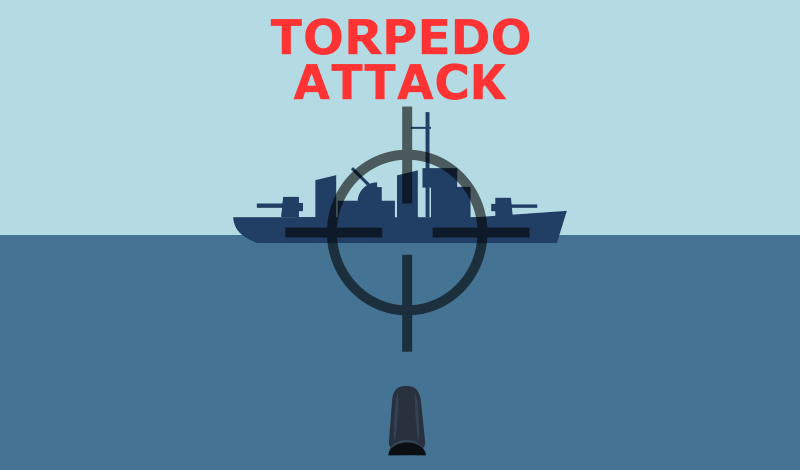}
	\caption{Space scientists are concerned with  debris or other hazardous objects approaching a satellite and destroying it. A missile (or torpedo) can quickly approach a target (or submarine) and destroy it. Source: https://images.google.com/. }
	\label{fig:satellite}
\end{figure}

In this paper, we propose a new class of stochastic processes to capture the stochastically growing nature of a process. We then obtain an optimal algorithm for detecting a change in distribution in this new class of processes.

\subsection{Satellite Application}
\label{sec:satellite}
In this section, we provide a brief discussion on satellite safety, the main motivation behind the quickest change detection problem studied in this paper. 

\begin{figure}[h]
	\centering
	\includegraphics[scale=0.25]{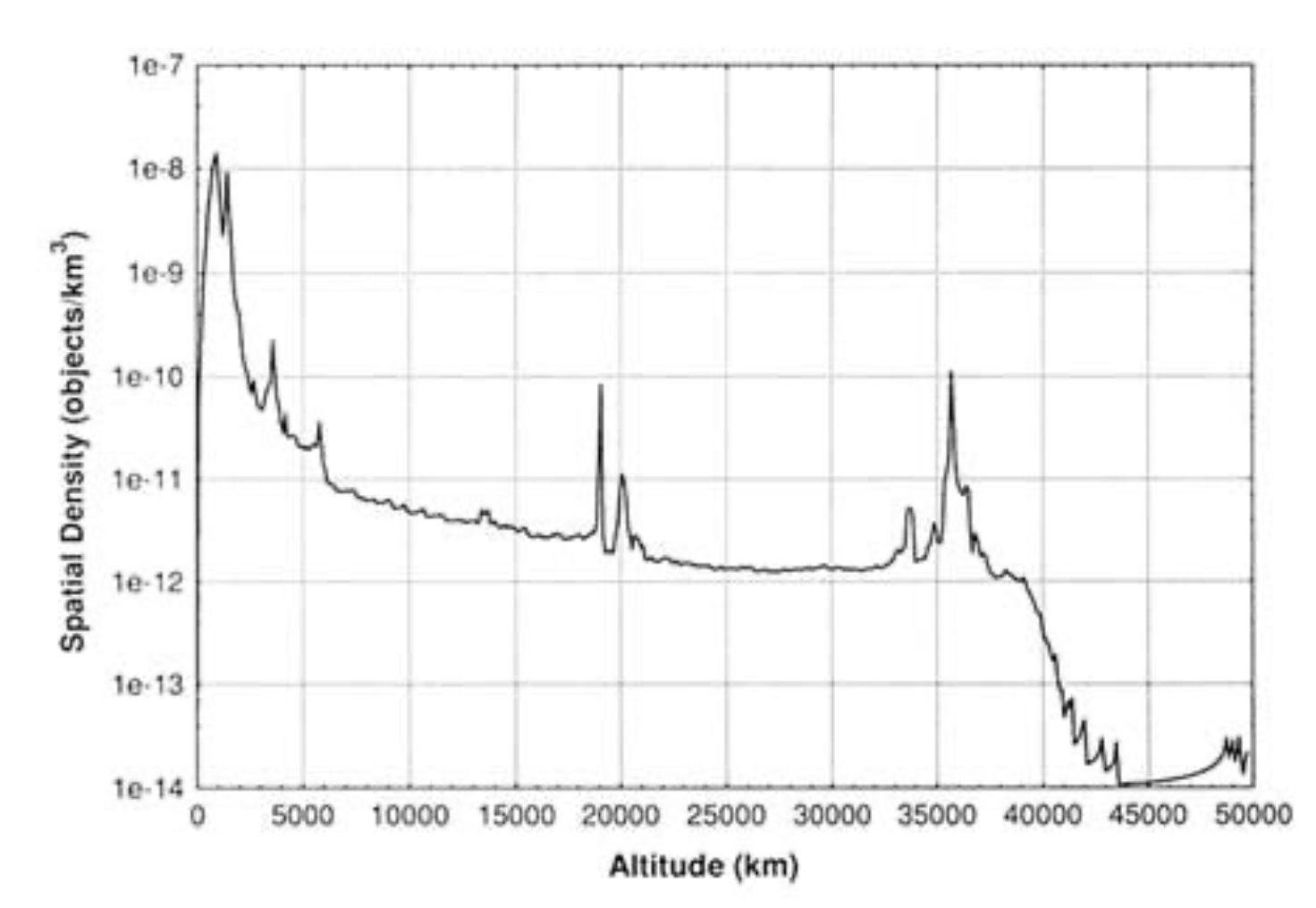}
	\includegraphics[scale=0.25]{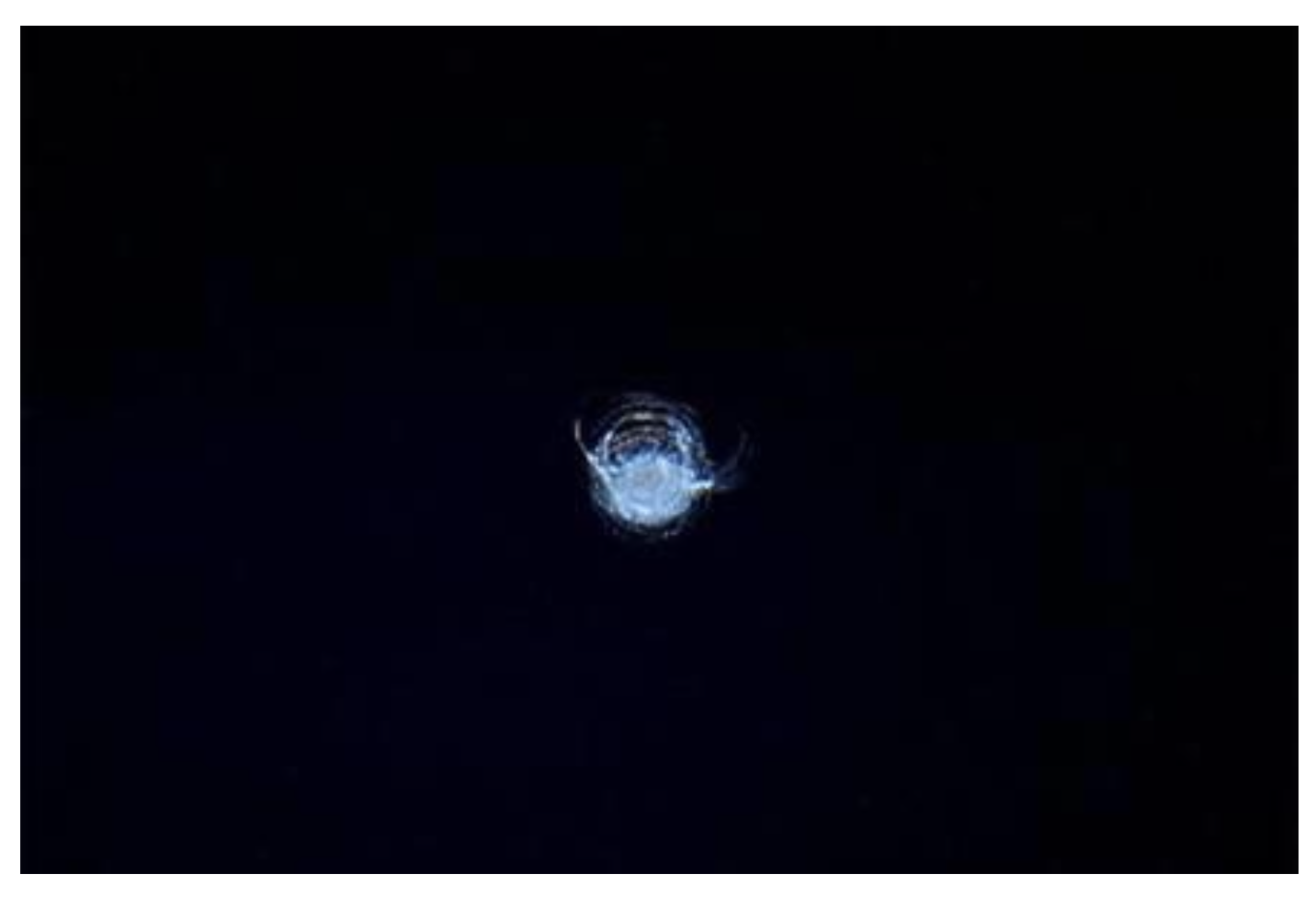}
	\includegraphics[scale=0.3]{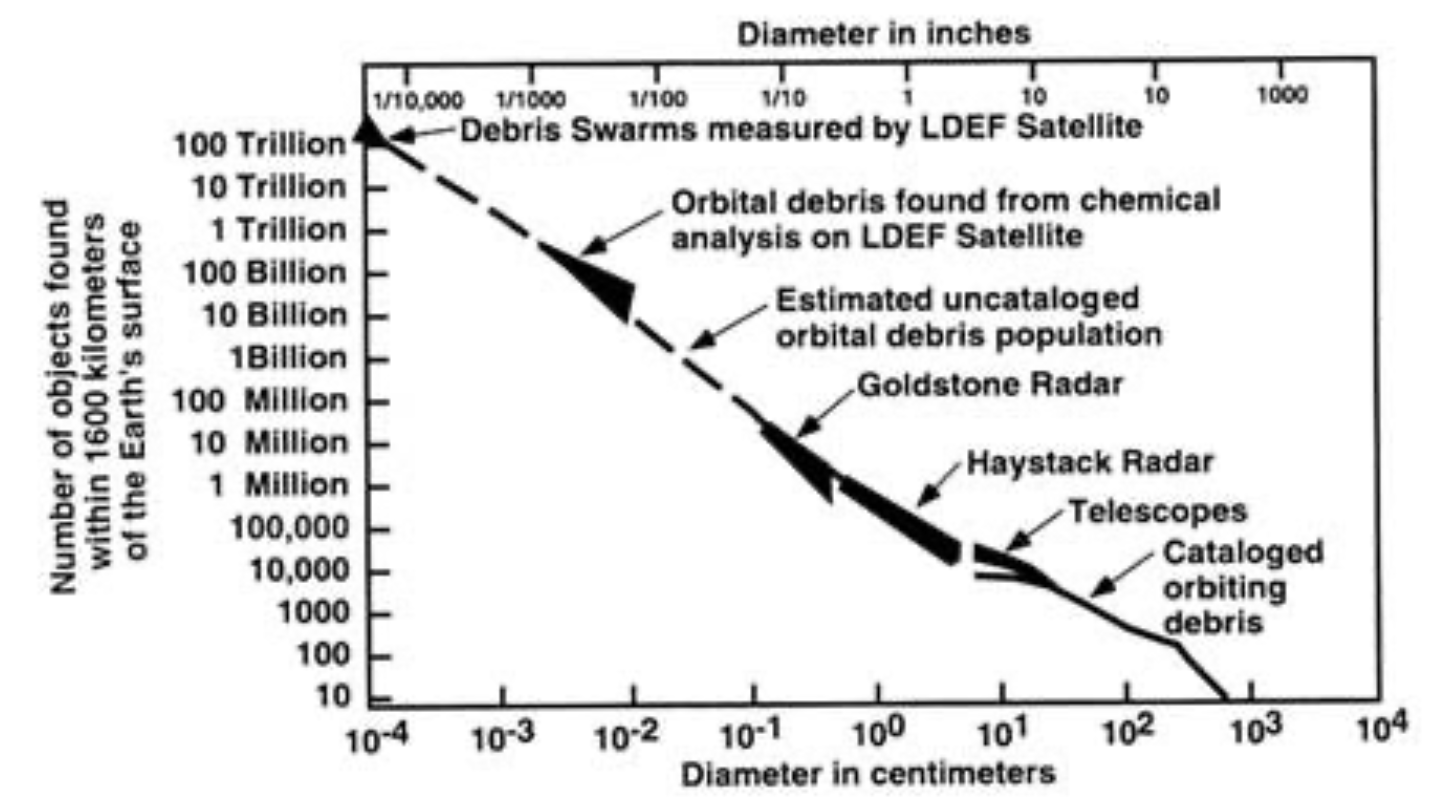}
	\caption{Left figure: Debris density in earth orbits; Right figure: Debris hit in ISS; Bottom figure: Debris size distribution.}
	\label{fig:debris}
\end{figure}

The problem of detection of the approaching object in space and tracking them finds many concrete applications as follows:
\begin{enumerate}
	\item In Fig.~\ref{fig:debris} (left and bottom),  we provide the size and density of objects present in Earth orbits. These objects are potentially hazardous for Astronauts doing Extra-Vehicular Activities (EVA) and can be detrimental to the space resources such as operational satellites, ISS. In Fig.~\ref{fig:debris} (right), we can see the impact of debris hit in ISS. In \cite{mohammed2022space, xiang2020space},  the authors have mapped the debris and detect them using on-board systems.
	\item Some of the major space accidents due to in space collisions and space debris are as follows:
	a)	A piece of space junk damaged the robotic arm of Lab within the ISS; A puncture in the thermal covering of robotic arm was observed.
	b)	The 1996 collision between the French Cerise with a catalogued debris from an Ariane rocket which tore the stabilization boom of the satellite.
	c)	Collision of Iridium-33 (a part of commercial communication network) with the Kosmos 2251, a first major satellite collision in space of 2 satellites in Earth Orbit on account of congestion and not deorbiting the defunct satellite. This resulted in more than thousand debris particles which could have potentially lead to the escalation of debris collision.
	d)	The recent collision in March 2021 between Yunhai-1 02 and debris from the Zenit-2 rocket created a debris of over 20 pieces.
	\item Debris: Artificial space debris along with the meteoroids are referred as MMOD (Micrometeoroids and Orbital Debris). These objects can cause sandblasting, especially to spacecraft appendages and optics that are difficult to be protected by any shield \cite{nasamedbrief}. Some of the types of space debris can be as follows:
	a)	spare space parts, 
	b)	deorbitted satellites, and 
	c)	launch vehicle parts.
	\item The orbital bodies such as functional and defunct satellites and their parts, lv parts, and the states of most of the catalogued objects are expected to evolve deterministically with the tolerances that are well modeled using the statistical approach illustrated in this paper. These reasonable assumptions are exploited to arrive at the useful results to enable autonomy in the future satellites. Thus the results in this paper can prove useful in Autonomous Meteors’ Avoidance Mechanism.
	\item This work envisages to firm-up the more realistic assumptions pertaining to situations such as debris characterised by Kessler syndrome and intelligent chaser systems that can be in space or on-ground and can be a potential threat to cause disaster to space resources.
	\item The immediate application of this work is to the autonomous spacecraft operations, which imparts substantial operational or commercial advantages to space industry. As the complexity of space operations is one of the major overheads associated with the satellite services, controlling this component can be critical in gaining an edge in the prevailing competition within the present space sector.
	\item One recent example can be observed when the ESA satellite required a manoeuver to avoid a likely collision with the SPACEX’s Starlink satellite.
\end{enumerate}

%



\section{Model and Problem Formulation}

\subsection{Data Model}
We introduce a new class of stochastic processes defined to model exploding nature of the post-change process: 

\medspace
\begin{definition}
	\label{def:exploding}
	We say that a process $\{X_n\}_{n \geq 1}$ with densities $\{f_n\}_{n \geq 0}$ is an {exploding process} if 
	\begin{enumerate}
		\item $\{X_n\}$ are jointly independent,
		\item $X_n \sim f_{n-1}$, $\forall n$, 
		\item $\frac{f_{n+1}(x)}{f_n(x)}$ is {increasing} in $x$, $\forall n$. We use $f_{n} \prec f_{n+1}$ to denote this {monotone likelihood ratio} ({MLR}) order. 
	\end{enumerate}
\end{definition}

If $f_{n+1}=f_n$, for all $n$, then we get an independent and identically distributed (i.i.d.) process.
Note that {MLR dominance implies }{stochastic dominance} \cite{krishnamurthy2016partially}: 
$$
\int_x^\infty f_n(x) dx \leq \int_x^\infty f_{n+1}(x) dx, \quad \forall x,n.  
$$

\subsection{Change point model} 

We assume that before change, data is i.i.d. with density $g \prec f_0$, and after change, an exploding process with sequence $\{f_n\}_{n \geq 0}$ of densities. Mathematically, there is a discrete-time $\nu$ such that 
\begin{equation*}
	X_n \sim 
	\begin{cases}
		g, &\quad \forall n < \nu, \\
		f_{n-\nu} &\quad \forall n \geq \nu.
	\end{cases}
\end{equation*}
Here we have used the notation $X \sim f$ to denote that the random variable $X$ has law $f$. Note that the post-change density of an observation depends on the location of the change point $\nu$. Our goal is to detect this change as quickly as possible, subject to a constraint on the rate of false alarms.

\subsection{Problem Formulation}
To solve the change detection problem, we are interested in two popular minimax formulations for the quickest change detection. To state the formulations, we define, for $1 \leq \nu \leq \infty$,	$\mathsf{E}_\nu$ as the expectation when the change occurs at time $\nu$. We consider the problem formulation of Lorden \cite{lord-amstat-1971}: 
\begin{equation}
	\label{eq:lorden}
	\begin{split}
		\min_\tau &\quad \sup_{\nu \geq 1} \; \text{ess sup} \; \mathsf{E}_\nu[(\tau - \nu + 1)^+| X_1, \dots, X_{\nu-1}], \\
		\text{subj. to }& \quad \mathsf{E}_\infty [\tau] \geq \gamma, 
	\end{split}
\end{equation}
where $\gamma$ is a constraint on the mean time to a false alarm. We will also consider the formulation of Pollak \cite{poll-astat-1985}:
\begin{equation}
	\label{eq:pollak}
	\begin{split}
		\min_\tau &\quad \sup_{\nu \geq 1} \; \mathsf{E}_\nu[\tau - \nu   |  \tau \geq \nu], \\
		\text{subj. to }& \quad \mathsf{E}_\infty [\tau] \geq \gamma, 
	\end{split}
\end{equation}
where again $\gamma$ is a constraint on the mean time to a false alarm.

\section{Candidate Algorithm: Exploding Cumulative Sum Algorithm}

We propose to use the following exploding Cumulative Sum (EX-CUSUM) statistic for change detection: 
\begin{equation}
	W_n = \max_{1 \leq k \leq n} \; \sum_{i=k}^n \; \log \frac{f_{i-k}(X_i)}{g(X_i)}. 
\end{equation}
The term $\sum_{i=k}^n \; \log \frac{f_{i-k}(X_i)}{g(X_i)}$ is the log-likelihood ratio of the observations between post-change and pre-change distributions, conditioned that the change occurs at time $k$. Since we do not know the change point, we take the maximum of all possible values at time $n$, i.e., $1 \leq k \leq n$. 

In the above statistic, note that the likelihood ratio of an observation $X_i$ at time $i$, 
$$
\frac{f_{i-k}(X_i)}{g(X_i)},
$$
depends on the relative distance $i-k$ between time $i$ and the hypothesis $k$ of the change point. 
It is because of this reason the EX-CUSUM statistic is not a special case of the generalized CUSUM statistic of Lai \cite{lai-ieeetit-1998}. 

To detect the change in law from i.i.d. with law $g$ to an exploding process $\{f_n\}$, we stop the first time the EX-CUSUM statistic is above a threshold $A$: 
\begin{equation}
	\tau_{ec} = \inf\{n \geq 1: W_n > A\}. 
\end{equation}
We select the threshold $A$ to control the rate of false alarms: the higher the threshold, the smaller the mean time to a false alarm $\Expect_\infty[\tau_{ec}]$ (this fact will be formally proved below). 

Our goal in this paper is to characterize conditions under which the EX-CUSUM algorithm is asymptotically optimal for Lorden's and Pollak's problems in \eqref{eq:lorden} and \eqref{eq:pollak}.

\section{Asymptotic Lower Bound on the Performance}
\subsection{Lai's Asymptotic Lower Bound}
In \cite{lai-ieeetit-1998}, Lai reported a general minimax theory for the quickest change detection. We first review it and discuss its limitations. 

It is assumed in \cite{lai-ieeetit-1998} that the pre-change densities are $f_0(X_i | X_1, \dots, X_{i-1})$ at time $i$ and the post-change densities are $f_1(X_i | X_1, \dots, X_{i-1})$ giving us the log-likelihood ratio
$$
Z_i = \log \frac{f_1(X_i | X_1, \dots, X_{i-1})}{f_0(X_i | X_1, \dots, X_{i-1})}.
$$
Note that the densities are general conditional densities allowing for data dependence but are not a function of the hypothesis on the change point (as defined in the paper). Lai showed that if $Z_i$s are such that there exists an information number $I > 0$ satisfying
\begin{equation}
	\begin{split}
		\lim_{n \to \infty} \; \sup_{\nu \geq 1} \; \text{ess sup } \mathsf{P}_\nu &\left(\max_{t \leq n} \sum_{i = \nu }^{\nu + t} Z_i \geq I(1+\delta)n  \; \; \bigg| \; X_1, \dots, X_{\nu-1}\right) = 0,
	\end{split}
\end{equation}
then we have the universal lower bound as $\gamma \to \infty$, 
\begin{equation}
	\begin{split}
		\min_{\tau} \; \sup_{\nu \geq 1} \; &\text{ess sup} \; \mathsf{E}_\nu[(\tau - \nu + 1)^+| X_1, \dots, X_{\nu-1}]\\
		&\geq \quad  \min_{\tau}  \; \sup_{\nu \geq 1} \; \mathsf{E}_\nu[\tau - \nu   |  \tau \geq \nu]\\
		& \quad \geq \quad \frac{\log \gamma}{I} (1+o(1)). 
	\end{split}
\end{equation}
Here the minimum over $\tau$ is over those stopping times satisfying $\Expect_\infty[\tau] \geq \gamma$. 
Lai further showed that under certain additional conditions on the $Z_i$s, the generalized CUSUM algorithm,
$$
\tau_c = \min\left\{n \geq 1: \max_{1 \leq k \leq n} \sum_{i=k}^n Z_i \geq \log(\gamma)\right\},
$$
is asymptotically optimal for both Lorden and Pollak's problems. 	 Specifically, 
$$
\mathsf{E}_\infty[\tau_c] \geq \gamma.
$$
Further, if $Z_i$s satisfy 
\begin{equation}
	\begin{split}
		\lim_{n \to \infty} \; \sup_{k \geq \nu \geq 1} \; \text{ess sup } \mathsf{P}_\nu & \left(\frac{1}{n}\sum_{i = k }^{k +n} Z_i \leq I - \delta \; \bigg| \; X_1, \dots, X_{k-1}\right) = 0,
	\end{split}
\end{equation}
then as $\gamma \to \infty$, $\tau_c$ achieves the lower bound:
\begin{equation}
	\begin{split}
		\sup_{\nu \geq 1} \; \text{ess sup} \; \mathsf{E}_\nu &[(\tau_c - \nu + 1)^+| X_1, \dots, X_{\nu-1}]  \leq  \frac{\log \gamma}{I}(1+o(1)), \quad \gamma \to \infty. 
	\end{split}
\end{equation}

It is not clear if these results are valid for the exploding process setting because in the latter setting likelihood ratios do depend on where the change point is. We discuss this next. 

\subsection{Sufficient Conditions for Change Point Dependent Likelihoods}

In this section, we extend Lai's results to the case of change point-dependent likelihoods. We continue to allow data dependence across time to state the more general result. Define the log-likelihood ratio at time $n$ when change occurs at $\nu$ as
$$
Z_{n, \nu} = \log \frac{f_{n,\nu}(X_n | X_1, \dots, X_{n-1})}{f_0(X_n|X_1, \dots, X_{n-1})}.
$$
\begin{theorem}
	\label{thm:modifiedconds}
	\begin{enumerate}
		\item Let there exist a positive number $I$ such that the bivariate log likelihood ratios $\{Z_{n, \nu}\}$ satisfy the following  condition: 
		\begin{equation}
			\label{eq:Znnu_LB}
			\begin{split}
				\lim_{n \to \infty} \; \sup_{\nu \geq 1} \; \text{ess sup } \mathsf{P}_\nu &\left(\max_{t \leq n} \sum_{i = \nu }^{\nu + t} Z_{i, \nu} \geq I(1+\delta)n \; \bigg| \; X_1, \dots, X_{\nu-1}\right) = 0.
			\end{split}
		\end{equation}
		Then, we have the universal lower bound as $\gamma \to \infty$, 
		\begin{equation}
			\begin{split}
				\min_{\tau} \; \sup_{\nu \geq 1} \; &\text{ess sup} \; \mathsf{E}_\nu[(\tau - \nu + 1)^+| X_1, \dots, X_{\nu-1}]\\
				&\geq \quad  \min_{\tau}  \; \sup_{\nu \geq 1} \; \mathsf{E}_\nu[\tau - \nu   |  \tau \geq \nu]\\
				& \quad \geq \quad \frac{\log \gamma}{I} (1+o(1)). 
			\end{split}
		\end{equation}
		Here the minimum over $\tau$ is over those stopping times satisfying $\Expect_\infty[\tau] \geq \gamma$. 
		\item The following modified generalized CUSUM algorithm,
		$$
		\tau_{mc} = \min\left\{n \geq 1: \max_{1 \leq k \leq n} \sum_{i=k}^n Z_{i,k} \geq \log(\gamma)\right\},
		$$
		satisfies
		$$
		\mathsf{E}_\infty[\tau_{mc}] \geq \gamma.
		$$
		\item Let the bivariate log-likelihood ratios $\{Z_{n, \nu}\}$ also satisfy 
		\begin{equation}
			\label{eq:Znnu_UB}
			\begin{split}
				\lim_{n \to \infty} \; \sup_{k \geq \nu \geq 1} \; \text{ess sup } \mathsf{P}_\nu & \left(\frac{1}{n}\sum_{i = k }^{k +n} Z_{i,k} \leq I - \delta \; \bigg| \; X_1, \dots, X_{k-1}\right) = 0.
			\end{split}
		\end{equation}
		Then as $\gamma \to \infty$, $\tau_{mc}$ achieves the lower bound:
		\begin{equation}
			\begin{split}
				\sup_{\nu \geq 1} \; \text{ess sup} \; \mathsf{E}_\nu &[(\tau_{mc} - \nu + 1)^+| X_1, \dots, X_{\nu-1}]  \leq  \frac{\log \gamma}{I}(1+o(1)), \quad \gamma \to \infty. 
			\end{split}
		\end{equation}
	\end{enumerate}
\end{theorem}
\begin{proof}
	The lower bound result goes through by replacing $Z_i$s by $Z_{n, \nu}$s in \cite{lai-ieeetit-1998}. This is because the proof relies on a change of measure argument to get the lower bound. Since the likelihood ratios here are a function of $Z_{n, \nu} $, the same change of measure argument works. 
	The proof of detection delay also goes through provided the above conditions are satisfied and everywhere $Z_i$ are replaced by $Z_{i,k}$. 
	It is not clear to the authors whether Lai's proof for the mean time to a false alarm result can be extended to the time-dependent setting. But, one can use another technique based on the Shiryaev-Roberts statistic and optional sampling theorem. Since the latter technique is classical \cite{tart-niki-bass-2014}, we skip the details. 
\end{proof}

While this paper was being written, the authors were made aware of another paper \cite{liang2021quickest} in which more general sufficient conditions for optimality (more general
than those in \cite{lai-ieeetit-1998}) have been established. In comparison with \cite{liang2021quickest}, we provide a novel way (using MLR order) to model stochastically growing nonstationary post-change process. In addition, our proof techniques for verifying these optimality conditions are slightly different and should be of independent interest. 
\medspace
\medspace

\section{Optimality of EX-CUSUM Algorithm}
In this section, we first simplify the conditions in Theorem~\ref{thm:modifiedconds} for exploding processes as defined in Definition~\ref{def:exploding}. We then provide additional comments on the simplified conditions to guarantee the optimality of the EX-CUSUM algorithm.

\subsection{Simplifying Lower Bound Condition for Exploding Processes}

In the theorem below, we show that the lower bound condition \eqref{eq:Znnu_LB} can be simplified in the case of exploding processes. 

\medspace
\medspace
\medspace
\begin{theorem}
	\label{thm:simplecondLB}
	To satisfy
	\begin{equation}
		\label{eq:Znnu_LB_2}
		\begin{split}
			\lim_{n \to \infty} \; \sup_{\nu \geq 1} \; \text{ess sup } \mathsf{P}_\nu &\left(\max_{t \leq n} \sum_{i = \nu }^{\nu + t} Z_{i, \nu} \geq I(1+\delta)n \; \bigg| \; X_1, \dots, X_{\nu-1}\right) = 0
		\end{split}
	\end{equation}
	for some $0 < I < \infty$,  it is sufficient that 
	$$
	\frac{1}{n} \sum_{k=1}^n Z_{k,1}  =  \frac{1}{n} \sum_{k=1}^n \log \frac{f_{k-1}(X_k)}{g(X_k)}\; \to \; I, \quad \text{a.s. under } \mathsf{P}_1. 
	$$
\end{theorem}
\begin{proof}
	Because of independence, we have
	\begin{equation*}
		\begin{split}
			\sup_{\nu \geq 1} \; &\text{ess sup } \mathsf{P}_\nu \left(\max_{t \leq n} \sum_{i = \nu }^{\nu + t} Z_{i,\nu} \geq I(1+\delta)n  \; \bigg| \; X_1, \dots, X_{\nu-1}\right) \\
			&=\sup_{\nu \geq 1} \mathsf{P}_\nu \left(\max_{t \leq n} \sum_{i = \nu }^{\nu + t} \log \frac{f_{i-\nu}(X_i)}{g(X_i)} \geq I(1+\delta)n \right).
		\end{split}
	\end{equation*}
	
	Since likelihood ratios are computed relative to the change point $\nu$, the probability on the right is not a function of $\nu$. This implies
	\begin{equation*}
		\begin{split}
			\sup_{\nu \geq 1} \; \mathsf{P}_\nu \left(\max_{t \leq n} \sum_{i = \nu }^{\nu + t} \log \frac{f_{i-\nu}(X_i)}{g(X_i)} \geq I(1+\delta)n \right) 
			&= \mathsf{P}_1 \left(\max_{t \leq n} \sum_{i = 1 }^{1 + t} \log \frac{f_{i-1}(X_i)}{g(X_i)} \geq I(1+\delta)n \right) \\
			&= \mathsf{P}_1 \left(\frac{1}{n}\max_{t \leq n} \sum_{i = 1 }^{1 + t} \log \frac{f_{i-1}(X_i)}{g(X_i)} \geq I(1+\delta) \right).
		\end{split}
	\end{equation*}
	Note that if the post-change process evolves differently for different change points, then the above simplification may not be true. Now, 
	if 
	$$\frac{1}{n} \sum_{i=1}^n \log \frac{f_{i-1}(X_i)}{g(X_i)}\; \to \; I, \quad \text{a.s. under} \; \mathsf{P}_1,$$ 
	then
	$$
	\frac{1}{n}\max_{t \leq n} \sum_{i = 1 }^{1 + t} \log \frac{f_{i,1}(X_i)}{g(X_i)} \to I, \quad \text{a.s. under} \; \mathsf{P}_1. 
	$$
	For proof of the above fact, see the Proof of Theorem 5.1 in \cite{bane-tit-2021}. 
	This proves the theorem because convergence almost surely implies convergence in probability. 
\end{proof}

\medspace
\medspace
We still need to characterize conditions under which
$$\frac{1}{n} \sum_{i=1}^n \log \frac{f_{i-1}(X_i)}{g(X_i)}\; \to \; I, \quad \text{a.s. under} \; \mathsf{P}_1,$$
for an exploding process. This will be done in Section~\ref{sec:Cantelli}.

\subsection{Controlling the False Alarm}
In this section, we show that by setting the threshold $A=\log \gamma$ in the EX-CUSUM algorithm, the constraints on the mean time to a false alarm can be satisfied. This is the content of the next theorem. 

Recall that 
$$
W_n = \max_{1 \leq k \leq n} \; \sum_{i=k}^n \; \log \frac{f_{i-k}(X_i)}{g(X_i)},
$$
$$
\tau_{ec} = \inf\{n \geq 1: W_n > A\}. 
$$

\begin{theorem}
	\label{thm:FAR}
	Setting $A = \log(\gamma)$ ensures that 
	$$
	\mathsf{E}_\infty[\tau_{ec}] \geq \gamma. 
	$$
\end{theorem}
\begin{proof} This proof technique is standard and has been used, for example, in \cite{tart-niki-bass-2014}. Because logarithm is monotonic and the maximum of positive quantities is always less than their sum, we have  
	\begin{equation*}
		\begin{split}
			\tau_{ec} &= \inf\left\{n \geq 1: \max_{1 \leq k \leq n} \; \sum_{i=k}^n \; \log \frac{f_{i-k}(X_i)}{g(X_i)}  > A\right\}\\
			&=\inf\left\{n \geq 1: \max_{1 \leq k \leq n} \; \prod_{i=k}^n \; \frac{f_{i-k}(X_i)}{g(X_i)}  > e^A\right\}\\
			&\geq \inf\left\{n \geq 1: \sum_{1 \leq k \leq n} \; \prod_{i=k}^n \; \frac{f_{i-k}(X_i)}{g(X_i)}  > e^A\right\} := \tau_{esr}. 
		\end{split}
	\end{equation*}

	The process
	$$
	R_n -n := \sum_{1 \leq k \leq n} \; \prod_{i=k}^n \; \frac{f_{i-k}(X_i)}{g(X_i)} - n
	$$
	is a $\mathsf{P}_\infty$-martingale. Assuming $\mathsf{E}[\tau_{esr}] < \infty$ (otherwise the false alarm constraint is trivially satisfied), we have
	\begin{equation*}
		\begin{split}
			\mathsf{E}\left[|R_{n} - n|; \{\tau_{esr} > n\}\right] &\leq \mathsf{E}\left[e^A + \tau_{esr}; \{\tau_{esr} > n\}\right] \\
			&\to 0, \quad n \to \infty.
		\end{split}
	\end{equation*}
	Thus, by Doob's optional sampling theorem \cite{chow-robb-sieg-book-1971},
	$$
	\mathsf{E}\left[R_{\tau_{esr}} - \tau_{esr} \right] = 0,
	$$
	and 
	$$
	\mathsf{E}\left[\tau_{esr} \right] = \mathsf{E}\left[R_{\tau_{esr}} \right] \geq e^A. 
	$$
	Now, set $A = \log \gamma$ to complete the proof. 
\end{proof}

\subsection{Simplifying Upper Bound Condition for Delay Analysis of an Exploding Process}
In this section, we simplify the condition \eqref{eq:Znnu_UB} for exploding processes. 
To complete this step, we need an intermediate result.
\begin{lemma}
	\label{lem:stocdom}
	Let $f(x_1, x_2, \dots, x_n)$ be a continuous function increasing in each of its arguments, with other arguments fixed. If $\{X_n\}$ is a stochastic process generated according to an exploding process, then for all $n, m, t$, 
	\begin{equation}
		\begin{split}
			\mathsf{P}& \left(f(X_n, X_{n+1}, \dots, X_{n+m}) \geq t\right)  \leq \mathsf{P}\left(f(X_{n+1}, X_{n+2}, \dots, X_{n+m+1}) \geq t\right).
		\end{split}
	\end{equation}
	
\end{lemma}
\begin{proof} The proof is similar to the one given in  \cite{unni-etal-ieeeit-2011}. Our proof requires an additional randomization step. 
\end{proof}

As in Theorem~\ref{thm:simplecondLB}, we show below that an almost sure condition is sufficient to satisfy the condition in \eqref{eq:Znnu_UB}. 

\medspace
\medspace
\medspace

\begin{theorem}
	To satisfy, for all $\delta > 0$,
	\begin{equation}
		\label{eq:Znnu_UB_2}
		\begin{split}
			\lim_{n \to \infty} \; \sup_{k \geq \nu \geq 1} \; \text{ess sup } \mathsf{P}_\nu & \left(\frac{1}{n}\sum_{i = k }^{k +n} Z_{i,k} \leq I - \delta \; \bigg| \; X_1, \dots, X_{k-1}\right) = 0
		\end{split}
	\end{equation}
	for some $0 < I < \infty$,  it is sufficient that the log-likelihoods are continuous and 
	$$
	\frac{1}{n} \sum_{k=1}^n Z_{k,1}  =  \frac{1}{n} \sum_{k=1}^n \log \frac{f_{k-1}(X_k)}{g(X_k)}\; \to \; I, \quad \text{a.s. under } \mathsf{P}_1. 
	$$
\end{theorem}
\begin{proof} Due to independence and the nature of exploding processes, we have
	\begin{equation*}
		\begin{split}
			\sup_{k \geq \nu \geq 1} \; &\text{ess sup } \mathsf{P}_\nu \left(\frac{1}{n}\sum_{i = k }^{k +n} Z_{i,k} \leq I - \delta \; \bigg| \; X_1, \dots, X_{k-1}\right) \\
			&=\sup_{k \geq \nu \geq 1} \mathsf{P}_\nu \left(\frac{1}{n}\sum_{i = k }^{k +n} \log \frac{f_{i-k}(X_i)}{g(X_i)} \leq I - \delta \right) \\
			&=\sup_{k \geq 1} \mathsf{P}_1 \left(\frac{1}{n}\sum_{i = k }^{k +n} \log \frac{f_{i-k}(X_i)}{g(X_i)} \leq I - \delta \right). 
		\end{split}
	\end{equation*}
	
	Because of Lemma~\ref{lem:stocdom}, the random variables
	$$
	\frac{1}{n}\sum_{i = k }^{k +n} \log \frac{f_{i-k}(X_i)}{g(X_i)}
	$$
	becomes stochastically bigger as $k$ increases. Thus, the maximum probability over $k$ is achieved at $k=1$. This gives
	\begin{equation*}
		\begin{split}
			\sup_{k \geq 1} \; &\mathsf{P}_1 \left(\frac{1}{n}\sum_{i = k }^{k +n} \log \frac{f_{i-k}(X_i)}{g(X_i)} \leq I - \delta \right) =\mathsf{P}_1 \left(\frac{1}{n}\sum_{i = 1 }^{1 +n} \log \frac{f_{i-1}(X_i)}{g(X_i)} \leq I - \delta \right).
		\end{split}
	\end{equation*}
	The last term goes to zero if 
	$$
	\frac{1}{n} \sum_{k=1}^n \log \frac{f_{k-1}(X_k)}{g(X_k)}\; \to \; I, \quad \text{a.s. under } \mathsf{P}_1. 
	$$
	This completes the proof. 
	
\end{proof}

Thus, for the optimality of the exploding CUSUM algorithms, it is enough to find conditions under which the almost sure convergence stated in the previous theorem is satisfied. 

\subsection{A Law of Large Numbers for Independent and Non-Identically Distributed Random Variables}
\label{sec:Cantelli}

In this section, we give conditions on exploding processes to guarantee 
$$
\frac{1}{n} \sum_{k=1}^n \log \frac{f_{k-1}(X_k)}{g(X_k)}\; \to \; I, \quad \text{a.s. under } \mathsf{P}_1. 
$$

We first recall Cantelli's strong law of large numbers. The proof can be found in \cite{shiryaev2019probability}.  
\begin{lemma}
	\label{lem:cantelli}
	Let $Y_1, Y_2, \dots$ be independent random variables with finite fourth moments, and let
	$$
	\mathsf{E}[|Y_n - \mathsf{E}[Y_n]|^4] \leq C, \quad n \geq 1,
	$$
	for some constant $C$. Then, as $n \to \infty$, 
	$$
	\frac{S_n - \mathsf{E}[S_n]}{n} \to 0, \quad \text{almost surely}. 
	$$
	Here $S_n = Y_1 + Y_2 + \dots Y_n$. 
\end{lemma}

Cantelli's strong law of large numbers provides us the needed tool to state our conditions. 

\medspace
\medspace
\medspace
\begin{theorem}
	\label{thm:KLDiverConv}
	We assume the following conditions hold for every $k \geq 1$:
	\begin{equation}
		\begin{split}
			\mathsf{E}&\left[\left(\log \frac{f_{k-1}(X_k)}{g(X_k)}\right)^4 \right] < \infty, \quad X_k \sim f_{k-1}, \\
			\mathsf{E}&\left[\left(\log \frac{f_{k-1}(X_k)}{g(X_k)} -  D(f_{k-1} \; \| \; g)\right)^4 \right] \leq C, \quad X_k \sim f_{k-1},
		\end{split}
	\end{equation}
	where $C$ is a constant. 
	Further, there exists an $I > 0$ such that
	\begin{equation}
		\begin{split}
			\frac{1}{n} &\sum_{k=1}^n D(f_{k-1} \; \| \; g)  \quad \to \quad I, \quad n \to \infty. 
		\end{split}
	\end{equation}
	Then, by Cantelli's strong law of large numbers (Lemma~\ref{lem:cantelli}),
	$$
	\frac{1}{n} \sum_{k=1}^n \log \frac{f_{k-1}(X_k)}{g(X_k)}\; \to \; I, \quad \text{a.s. under } \mathsf{P}_1. 
	$$
	
\end{theorem}

\medspace
\medspace
\medspace
\subsection{Asymptotic Optimality of EX-CUSUM Algorithm}
The previous theorem provides further simplification on the conditions needed for the optimality of the EX-CUSUM algorithm. We state this as a theorem:

\medspace
\medspace
\begin{theorem} 
	Under moment conditions stated in Theorem~\ref{thm:KLDiverConv}, if there exists an $I > 0$, such that
	$$
	\frac{1}{n} \sum_{k=1}^n D(f_{k-1} \; \| \; g)  \quad \to \quad I, \quad n \to \infty, 
	$$
	then the EX-CUSUM algorithm is asymptotically optimal for both Lorden's and Pollak's minimax formulations. 
\end{theorem}

\subsection{Gaussian Example}
\label{sec:Gaussexample}
We now give an example of an exploding process model for which all the conditions stated in this paper are satisfied. 
We assume that 
$$
g =  \mathcal{N}(0,1). 
$$
and 
$$
f_n = \mathcal{N}(\mu_n, 1)
$$
with
$$
0 \leq \mu_n \uparrow \mu. 
$$

All the likelihood ratios in this example are monotone because $0 \leq \mu_n \to \mu$. Also,  log-likelihood ratios are continuous because they are linear. The fourth-moment condition on log-likelihood ratios is satisfied because of Gaussianity.

Further, the log-likelihood ratio between $f_n$ and $g$ is given by
$$
\log \frac{f_n(X_{n+1})}{g(X_{n+1})} = \mu_n (X_{n+1} - \frac{\mu_n}{2}). 
$$
Thus,
\begin{equation}
	\begin{split}
		D(f_n \; \| \; g) &= \Expect_{X_{n+1} \sim f_n} \left[\log \frac{f_n(X_{n+1})}{g(X_{n+1})} \right] = \mu_n \left[ \Expect_{X_{n+1} \sim f_n}[X_{n+1}] - \frac{\mu_n}{2}\right]\\
		&=\mu_n \left(\mu_n - \frac{\mu_n}{2}\right) = \frac{\mu_n^2}{2}. 
	\end{split}
\end{equation}
This implies
$$
D(f_n \; \| \; g) = \frac{\mu_n^2}{2} \to \frac{\mu^2}{2}. 
$$
This further implies 
$$
\frac{1}{n} \sum_{k=1}^n D(f_{k-1} \; \| \; g)  \quad \to \quad \frac{\mu^2}{2}
$$
by Cesaro sum limit. 

Finally, we need to check if the fourth central moments of the log-likelihoods are uniformly bounded. Let
$$
\xi_n := \log \frac{f_n(X_{n+1})}{g(X_{n+1})} = \mu_n (X_{n+1} - \frac{\mu_n}{2}). 
$$
Also, let
$$
\Xi_n := \Expect_{X_{n+1} \sim f_n}[\xi_n] = \frac{\mu_n^2}{2}. 
$$
Then,
\begin{equation}
	\begin{split}
		\Expect_{X_{n+1} \sim f_n}\left[(\xi_n-\Xi_n)^4 \right]
		   &= \Expect_{X_{n+1} \sim f_n} \left[\left(\mu_n (X_{n+1} - \frac{\mu_n}{2}) - \frac{\mu_n^2}{2}\right)^4\right] \\
		& \quad  = \Expect_{X_{n+1} \sim f_n} \left[\left(\mu_n X_{n+1} - \mu_n^2\right)^4\right] \\
		& \quad  = \Expect_{X_{n+1} \sim f_n} \left[\left(\mu_n( X_{n+1} - \mu_n)\right)^4\right] \\
		&\quad= \mu_n^4 \Expect_{X_{n+1} \sim f_n} \left[\left(( X_{n+1} - \mu_n)\right)^4\right]\\
		&\quad \leq 3\mu^4. 
	\end{split}
\end{equation}
The last inequality is true because $\mu_n \uparrow \mu$ and because under $f_n$, $X_{n+1} - \mu_n \sim \mathcal{N}(0,1)$, and the latter's fourth moment is exactly equal to $3$. Thus, the fourth central moments are uniformly bounded by $C=3 \mu^4$. All the above arguments show that the EX-CUSUM algorithm is asymptotically optimal for this example.

\section{Numerical Results}
In this section, we apply the EX-CUSUM algorithm to the Gaussian data example discussed in Section~\ref{sec:Gaussexample}. To generate data in Fig.~\ref{fig:numGaussdata}, we used 
$$
\mu_n = \arctan(n) \to \mu = \frac{\pi}{2}. 
$$
As seen in the figure, the EX-CUSUM statistic stays close to zero before the change point of $80$ and grows toward $\infty$ after the change. This growth can be detected using a well-designed threshold. 
\begin{figure}[h]
	\centering
	\includegraphics[scale=0.4]{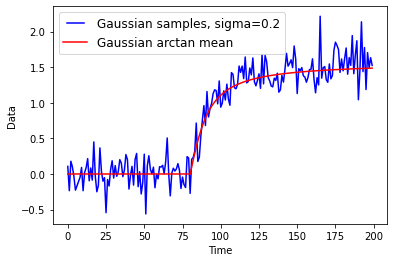}
	\includegraphics[scale=0.4]{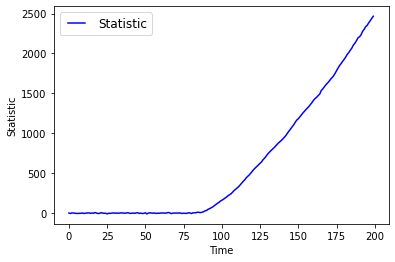}\\
	\caption{EX-CUSUM algorithm applied to Gaussian data. }
	\label{fig:numGaussdata}
\end{figure}

\section{Conclusions}
We introduced a new class of stochastic processes to model exploding nature of post-change observations in some change-point problems. Such observations are common in satellite and military applications where an approaching enemy object can cause the observation to grow stochastically over time. We proved that under mild conditions on the growth of Kullback-Leibler divergence, our proposed  EX-CUSUM algorithm is asymptotically optimal, as the mean time to false alarm grows to infinity. In our future work, we will investigate the exact optimality of the algorithm and also obtain computationally efficient methods for detecting changes. 

\section*{Acknowledgment}
A part of this work was presented at the 58th Annual Allerton Conference on Communication, Control, and Computing, Monticello, IL, September 2022.  The work of Taposh Banerjee was partially supported by the National Science Foundation under Grant 1917164 through a subcontract from Vanderbilt University. Tim Brucks passed away while we were working on this paper. This paper is dedicated to his memory.

\bibliographystyle{tfs}
\bibliography{TaposhQCD}

\begin{thebibliography}{10}
\providecommand{\MR}{\relax\unskip\space MR }
\providecommand{\url}[1]{\normalfont{#1}}
\providecommand{\urlprefix}{Available at }

\bibitem{nasamedbrief}
\emph{The threat of orbital debris and protecting nasa space assets from
  satellite collisions},
  \url{http://images.spaceref.com/news/2009/ODMediaBriefing28Apr09-1.pdf}.
  Accessed: 2023-02-03.

\bibitem{bane-tit-2021}
T. Banerjee, P. Gurram, and G.T. Whipps, \emph{A {Bayesian} theory of change
  detection in statistically periodic random processes}, IEEE Transactions on
  Information Theory 67 (2021), pp. 2562--2580.

\bibitem{chow-robb-sieg-book-1971}
Y.S. Chow, H. Robbins, and D. Siegmund, \emph{Great expectations: the theory of
  optimal stopping}, Houghton Mifflin, 1971.

\bibitem{krishnamurthy2016partially}
V. Krishnamurthy, \emph{Partially Observed Markov Decision Processes},
  Cambridge University Press, 2016.

\bibitem{lai-ieeetit-1998}
T.L. Lai, \emph{Information bounds and quick detection of parameter changes in
  stochastic systems}, IEEE Trans. Inf. Theory 44 (1998), pp. 2917 --2929.

\bibitem{liang2021quickest}
Y. Liang, A.G. Tartakovsky, and V.V. Veeravalli, \emph{Quickest change
  detection with non-stationary post-change observations}, To appear in IEEE
  Transactions on Information Theory  (2022).

\bibitem{lord-amstat-1971}
G. Lorden, \emph{Procedures for reacting to a change in distribution}, Ann.
  Math. Statist. 42 (1971), pp. 1897--1908.

\bibitem{mohammed2022space}
S. Mohammed, M.J.S. Teja, P. Kora, D.E. Vunnava, and M. Lokesh, \emph{Space
  Debris Detection Unit for Spacecrafts}, in \emph{2022 International
  Conference on Computer Communication and Informatics (ICCCI)}. IEEE, 2022,
  pp. 1--2.

\bibitem{mous-astat-1986}
G.V. Moustakides, \emph{Optimal stopping times for detecting changes in
  distributions}, Ann. Statist. 14 (1986), pp. 1379--1387.

\bibitem{Pergamenchtchikov2018}
S. Pergamenchtchikov and A.G. Tartakovsky, \emph{Asymptotically optimal
  pointwise and minimax quickest change-point detection for dependent data},
  Statistical Inference for Stochastic Processes 21 (2018), pp. 217--259.

\bibitem{poll-astat-1985}
M. Pollak, \emph{Optimal detection of a change in distribution}, Ann. Statist.
  13 (1985), pp. 206--227.

\bibitem{poor-hadj-qcd-book-2009}
H.V. Poor and O. Hadjiliadis, \emph{Quickest detection}, Cambridge University
  Press, 2009.

\bibitem{shiryaev2019probability}
A.N. Shiryaev, \emph{Probability-2}, Vol.~95, Springer, 2019.

\bibitem{shir-opt-stop-book-1978}
A.N. Shiryayev, \emph{Optimal Stopping Rules}, Springer-Verlag, New York, 1978.

\bibitem{tart-niki-bass-2014}
A.G. Tartakovsky, I.V. Nikiforov, and M. Basseville, \emph{Sequential Analysis:
  {Hypothesis} Testing and Change-Point Detection}, Statistics, CRC Press,
  2014.

\bibitem{tart-veer-siamtpa-2005}
A.G. Tartakovsky and V.V. Veeravalli, \emph{General asymptotic {Bayesian}
  theory of quickest change detection}, SIAM Theory of Prob. and App. 49
  (2005), pp. 458--497.

\bibitem{tartakovsky2017asymptotic}
A.G. Tartakovsky, \emph{On asymptotic optimality in sequential changepoint
  detection: Non-iid case}, IEEE Transactions on Information Theory 63 (2017),
  pp. 3433--3450.

\bibitem{unni-etal-ieeeit-2011}
J. Unnikrishnan, V.V. Veeravalli, and S.P. Meyn, \emph{Minimax robust quickest
  change detection}, IEEE Trans. Inf. Theory 57 (2011), pp. 1604 --1614.

\bibitem{veer-bane-elsevierbook-2013}
V.V. Veeravalli and T. Banerjee, \emph{Quickest Change Detection}, Academic
  Press Library in Signal Processing: Volume 3 -- Array and Statistical Signal
  Processing, 2014.

\bibitem{xiang2020space}
Y. Xiang, J. Xi, M. Cong, Y. Yang, C. Ren, and L. Han, \emph{Space debris
  detection with fast grid-based learning}, in \emph{2020 IEEE 3rd
  International Conference of Safe Production and Informatization (IICSPI)}.
  IEEE, 2020, pp. 205--209.

\end{thebibliography}

\end{document}